\documentclass[10pt]{article}

\usepackage[margin=1in]{geometry}
\usepackage{natbib}
\usepackage{parskip}
\usepackage{amsmath}
\usepackage{algorithm, algorithmicx, algpseudocode}
\usepackage{amssymb}
\usepackage{graphicx}
\usepackage{mathtools}
\usepackage{needspace} 
\usepackage{multicol}
\usepackage{listings}
\usepackage{amsthm}
\usepackage[shortlabels]{enumitem}
\usepackage[italicdiff]{physics}
\usepackage{cancel}
\usepackage[dvipsnames]{xcolor}
\usepackage{verbatim}
\usepackage{comment}
\usepackage{array}
\usepackage{tikz-cd}
\usepackage{import}
\usepackage{booktabs}
\usepackage{comment}
\usepackage{array}
\usepackage[colorlinks, linkcolor=blue, citecolor=ForestGreen]{hyperref}
\usepackage{forest}
\usepackage{float}
\usepackage{xspace}
\usepackage{wrapfig}
\usepackage{thm-restate}
\usepackage{adjustbox}
\usepackage{bm}
\usepackage{bbm}
\usepackage{complexity}
\usepackage{makecell}
\usepackage{nicefrac}
\usepackage{multirow}
\usetikzlibrary{calc, positioning, external, arrows.meta}

\allowdisplaybreaks

\DeclareMathOperator*{\argmin}{argmin}
\DeclareMathOperator*{\argmax}{argmax}

\let\cite\citep

\newfloat{protocol}{tbp}{lop}
\floatname{protocol}{Protocol}

\makeatletter
\renewcommand\thmt@autorefsetup{%
  \ifthmt@isnumbered
  \fi
}
\makeatother

\newtheorem{theorem}{Theorem}
\numberwithin{theorem}{section}
\newtheorem*{theorem*}{Theorem}

\newtheorem{lemma}[theorem]{Lemma}

\newtheorem{remark}[theorem]{Remark}

\theoremstyle{definition}

\usepackage{caption}
\makeatletter
\def\thanks#1{\protected@xdef\@thanks{\@thanks
        \protect\footnotetext{#1}}}
\makeatother

\title{Online Resource Allocation with Replenishable Budgets}
\author{
    Eleonora Fidelia Chiefari \\
    \texttt{eleonorafidelia.chiefari@polimi.it} \\
    Politecnico di Milano 
\and
    Francesco Emanuele Stradi \\
    \texttt{francescoemanuele.stradi@polimi.it} \\
    Politecnico di Milano 
\and
    Alberto Marchesi \\
    \texttt{alberto.marchesi@polimi.it} \\
    Politecnico di Milano
}
\date{\today}

\begin{document}

\maketitle

\begin{abstract}\noindent
  \emph{Online Resource Allocation} (ORA) is a fundamental framework for sequential decision-making problems under budget constraints. Classical ORA models typically assume that resources are monotonic, meaning that selecting actions can only decrease the available budget. In this work, we study a more general setting with \emph{replenishable} budgets, in which actions may either consume or replenish resources over time. This extension is necessary to capture scenarios such as inventory systems or energy markets in which capacity can be actively recovered.
We develop a dual-based algorithm that recovers best-of-both-worlds guarantees for standard ORA when the replenishment factor $\beta = 0$, and improves them when $\beta > 0$. In particular, our algorithm attains $\widetilde{\mathcal O}(\sqrt{T})$ regret in the stochastic setting and $\widetilde{\mathcal O}(\sqrt{T})$ $\alpha$-regret in the adversarial setting, where $\alpha$ depends on the per-round budget and on the replenishment factor of the void action. Moreover, the algorithm ensures strict satisfaction of the budget constraints. 
\end{abstract}

\newpage

\tableofcontents

\newpage

\section{Introduction} 

\emph{Online Resource Allocation} (ORA) \citep{devanur2009the, feldman2010} is a fundamental framework for sequential decision-making under \emph{budget constraints}. At each round, the decision maker observes a request and selects an appropriate action. Crucially, before making this decision, the request fully reveals the reward and resource consumption associated with every possible action \citep{agrawal_nearoptimal, KesselheimRTV13}. This full-information structure distinguishes ORA from most online learning paradigms, in which the learner must act without complete knowledge of the current-round outcomes. Consequently, the main challenge in ORA stems from uncertainty about future requests: learning algorithms must balance the immediate rewards obtained from current requests against the need to preserve resources for future ones \citep{devanur_nearoptimal}.

ORA problems arise in a wide range of applications, including online auctions \citep{balseiro2019learning}, online advertising \citep{mehta} and revenue management \citep{ball2009toward}. Driven by these practical motivations, various works in literature have addressed stochastic ORA problem \citep{li2020, sun2022nearoptimalprimaldualalgorithmsquantitybased}. A few efforts have been made beyond standard i.i.d assumption, considering non-stationary setting \citep{esfandiari} and adversarial setting \citep{balseiro2019learning, stradino}, although under strict assumptions. 
Notably, \citet{balseiro2023best} are the first to consider general ORA problems in both stochastic and adversarial regimes by designing a \emph{best-of-both-world} algorithm. 

Despite their contributions, the vast majority of the existing ORA literature relies on a \emph{monotonic} consumption model: every action taken by the agent depletes the available budget. The budget decreases over time until it is eventually exhausted, and the only action left is a \emph{void} action, which has no effect on the resources. While natural in many classical applications, this assumption excludes a broad class of sequential allocation problems in which resources can also be replenished. In many real-world scenarios, actions may either consume or restore capacity, and the agent may resort to a budget-restoring action when it is running low on budget. Consider, for instance, an inventory system in which a retailer can satisfy customer requests depleting the stock, or place orders from suppliers replenishing the stock. 
Similar dynamics arise in cloud computing, when a server can allocate computational resources to incoming tasks and at the same time terminate tasks to recover memory and CPU capacity. 

In such settings, resource dynamics are inherently \emph{non-monotonic}. The budget oscillates over time, introducing new algorithmic challenges that standard ORA frameworks cannot handle. Classical primal-dual algorithms~\citep{balseiro2023best} rely on monotonic resource depletion to naturally bound the dual variables. When costs can be negative due to replenishments, standard gradient updates risk to artificially decrease the dual multipliers and make the algorithm over-consume resources. 

In this work, we address the following research question:
\begin{quote}
    \centering 
    \emph{How does the replenishment factor affect the achievable regret bounds and the competitive ratio in online resource allocation?}
\end{quote}
We answer this question by proposing a novel best-of-both-worlds primal-dual framework capable of handling non-monotonic resource utilization. We prove that the availability of replenishments can be structurally exploited to improve the agent's performance. Specifically, our proposed algorithm achieves simultaneously sublinear expected regret in stochastic environments and sublinear fractional regret in adversarial ones, while deterministically satisfying the budget constraints.

\subsection{Original Contribution}
In this work, we formally study the online resource allocation problem with replenishable budgets. We propose a dual-based framework capable of handling non-monotonic resource utilization under both stochastic and adversarial assumptions. 
As summarized in Table~\ref{table-comparison}, our algorithm recovers the optimal guarantees of the state-of-the-art \citep{balseiro2023best} for standard monotonic environments ($\beta = 0$) and improves state-of-the-art bounds when replenishment is active ($\beta > 0$). Specifically, we establish the following theoretical guarantees: 
\begin{itemize}
    \item $\widetilde{\mathcal O}(\sqrt{T})$ cumulative expected regret in the \emph{stochastic} setting, i.e., when requests are sampled i.i.d from an unknown fixed distribution. Notably, the bound scales inversely with the augmented margin $\rho + \beta$, yielding a strictly tighter regret guarantee compared to standard monotonic models; 
    \item $\widetilde{\mathcal O}(\sqrt{T})$ cumulative $\alpha$-regret in the \emph{adversarial} setting, when requests are selected in advance by an adversary, with $\alpha \coloneqq \nicefrac{\rho + \beta}{1 + \beta}$ problem-dependent parameter. 
\end{itemize}
The bounds in Table~\ref{table-comparison} use knowledge of $\beta$ to tune the learning rate. At the same time, in both settings, the algorithm strictly satisfies the hard budget constraints over the time horizon. Furthermore, the execution of the algorithm does not require prior knowledge of the true replenishment factor $\beta$. 

\subsection{Technical Challenges}
In classical ORA, primal-dual algorithms rely on the monotonic depletion of resources to naturally bound and stabilize the dual variables \citep{balseiro2023best}. Since resource consumption is always non-negative, the dual update acts as a straightforward penalty mechanism. Specifically, when per-round consumption exceeds $\rho$, the corresponding dual multipliers increase, discouraging the primal algorithm from selecting resource-intensive actions. 

When introducing active replenishment, this monotonic structure breaks down. Negative ``costs'' decrease the corresponding dual multipliers. An action might consume one resource while actively restoring another, pushing the respective dual multipliers in opposite directions. This lack of alignment creates a critical vulnerability in the Lagrangian objective: cross-constraints compensation. An action that severely depletes a critical resource may still appear optimal if it heavily replenishes a different one, masking the violation. Thus, the algorithm risks greedily over-consuming the scarce resource. This mechanism leads to premature budget exhaustion, forcing the agent to play the void action until the budget is recovered. 

Overcoming this limitation requires a dual-space approach that tolerates non-monotonic consumption and unaligned constraints. In this paper, we tackle this challenge by designing a best-of-both-worlds dual-based framework that prevents the artificial deflation of lagrangian multipliers, keeping them tightly bounded even in the presence of negative costs. 

\begin{table*}[t!]
    \caption{Comparison of theoretical guarantees, with known $\beta$ for our bounds. Exposing the dependence on the budget margin ($\rho$) and the replenishment margin ($\beta$) highlights how our framework not only accommodates non-monotonic ORA ($\beta > 0$), but actively exploits replenishments to tighten the regret bounds and significantly improve the adversarial competitive ratio. \\}
    \label{table-comparison}
    \centering
    \small
    {\renewcommand{\arraystretch}{1.6}
        \setlength{\tabcolsep}{12pt}
        \begin{tabular}{llcc}
            \toprule
             \textbf{Metric} 
             & \textbf{Margin} 
             & \textbf{\citet{balseiro2023best}} 
             & \textbf{Our Work} \\
            \midrule
            
            \multirow{2}{*}{Stochastic regret ($R_T$)} 
             & $\bm{\beta = 0}$ 
             & $\widetilde{\mathcal{O}}\left(\frac{1}{\rho}\sqrt{T}\right)$ 
             & $\widetilde{\mathcal{O}}\left(\frac{1}{\rho}\sqrt{T}\right)$ \\
             
             & $\bm{\beta > 0}$ 
             & -- 
             & $\widetilde{\mathcal{O}}\left(\frac{1}{\rho + \beta}\sqrt{T}\right)$ \\
            
            \midrule
             
            \multirow{2}{*}{Adversarial fractional regret ($\alpha$-$R_T$)} 
             & $\bm{\beta = 0}$ 
             & $\widetilde{\mathcal{O}}\left(\frac{1}{\rho}\sqrt{T}\right)$ 
             & $\widetilde{\mathcal{O}}\left(\frac{1}{\rho}\sqrt{T}\right)$ \\
             
             & $\bm{\beta > 0}$ 
             & -- 
             & $\widetilde{\mathcal{O}}\left(\frac{1}{\rho + \beta}\sqrt{T}\right)$ \\
            
            \midrule
             
            \multirow{2}{*}{Competitive Ratio ($\alpha$)} 
             & $\bm{\beta = 0}$ 
             & $\rho$ 
             & $\rho$ 
             \cr
             & $\bm{\beta > 0}$ 
             & -- 
             & $\frac{\rho + \beta}{1 + \beta}$ \\
             
            \bottomrule
        \end{tabular}
    }
\end{table*}

\subsection{Related works} 
\paragraph{Online Resource Allocation} Classical Online Resource Allocation framework has been studied in many real-world scenarios, such as online auctions~\citep{zhou,balseiro2019learning}, online advertising~\citep{mehta}, and revenue management~\citep{talluri2004theory, ball2009toward}. Most of the existing literature studied ORA under stochastic feedback. Early two-phases approaches~\citep{devanur2009the, feldman2010} estimated static dual prices from an initial data sample, leading to suboptimal regret bounds of order $\mathcal{O}(T^{2/3})$. To overcome this limitation, subsequent works~\citep{agrawal_nearoptimal, KesselheimRTV13, devanur_nearoptimal} dynamically updated decisions by periodically re-solving linear programming relaxations, securing the optimal $\widetilde{\mathcal{O}}(\sqrt{T})$. To avoid the computational complexity of repeated LP solvers, recent studies have shifted toward gradient-based primal-dual schemes~\citep{li2020, sun2022nearoptimalprimaldualalgorithmsquantitybased}. Although these methods attain strong theoretical guarantees, they require specific assumptions or problem structures.

When relaxing the i.i.d. assumption, the literature branches into non-stationary environments~\citep{Ciocan2011DynamicAP, esfandiari} and fully adversarial models. Early results in the adversarial setting study the specific problem of online matching and AdWords~\citep{buchbinder, mehta}, without a direct application to more general settings. The first \emph{best-of-both-worlds} algorithm for Online Resource Allocation is the one presented in \citep{balseiro2023best}, which is a dual-based approach simultaneously $\widetilde{\mathcal{O}}(\sqrt{T})$ regret in the stochastic setting,and $\widetilde{\mathcal{O}}(\sqrt{T})$ $\alpha$-regret in the non-stationary and in the adversarial setting. \cite{chiefari2026onlineresourceallocationgeneral} study a generalized ORA framework, finding best-of-both-worlds guarantees incorporating general long-term constraints alongside standard budgets. The algorithm proposed in this work attains $\widetilde{\mathcal{O}}(\sqrt{T})$ regret in the stochastic setting, $\widetilde{\mathcal{O}}(\sqrt{T})$ $\alpha$-regret in the adversarial setting and $\widetilde{\mathcal{O}}(\sqrt{T})$ cumulative violation of the general constraints in both settings.

Parallel analytical efforts have been dedicated to more specific domains. \citep{balseiro2019learning} propose a pacing-based strategy attaining $\widetilde{\mathcal{O}}(\sqrt{T})$ regret in the stochastic setting under suitable assumptions, and provides an impossibility result for the adversarial case: no algorithm can exceed the competitive ratio $\nicefrac{v}{\rho}$, with $v$ denoting an upper bound on the bidder's per-item value. A later work~\citep{feng2023onlinebiddingalgorithmsreturnonspend} extends the online bidding literature by including the Return-on-Spend constraint, which ensures that the ratio of value to expenditure remains over a prescribed threshold. Although it obtains strong guarantees, this analysis relies on the specific structure of the RoS constraints, thus cannot be generalized. Beyond these bidding-specific results, other contributions have pursued competitive-ratio guarantees directly under budget constraints. For instance, in the single-resource online Knapsack setting, \citep{zhou} establish competitive ratios depending on the range of the value-to-weight ratio of the items. 

Despite their theoretical depth, all the aforementioned frameworks inherently rely on a strictly monotonic budget depletion model. Our work fundamentally diverges from this paradigm. By modeling actions that may yield negative costs, we introduce a non-monotonic budget dynamic through active replenishments, requiring a self-regulating mechanism in the dual space. This concept is detailed in Section~\ref{lambda}.

\paragraph{Online Learning with Replenishment} 
The idea that resource consumption need not be monotonic has surfaced only recently in the online learning literature, and almost exclusively within the bandit feedback model. The first step in this direction is due to \citep{kumar2022nonmonotonicresourceutilizationbandits}, who introduced the \emph{Bandits with Replenishable Knapsacks} (BwRK) framework. Their work addresses non-monotonic resource utilization by providing instance-dependent logarithmic bounds specifically tailored for stochastic input models, extending the classic BwK framework \citep{Badanidiyuru_2013, castiglioni2022online}. Building on this formulation, \citep{bernasconi2024bandits} proposed a primal-dual method with best-of-both-worlds guarantees for the replenishable setting. Their algorithm yields a $\widetilde{\mathcal{O}}(\sqrt{T})$ regret bound in the stochastic setting, matching previous results, and ensures a constant competitive ratio in the adversarial setting. 

While these works successfully tackle the challenges of non-monotonic resource consumption, they remain confined to the bandit feedback framework, where the agent observes the realized reward and resource consumption \emph{after} the action is selected. Our work addresses the challenge of active replenishment within the full-information ORA setting. We provide further discussion on related works in Appendix~\ref{rw}

\section{Preliminaries}\label{sec:preliminaries}

We study the problem of \emph{Online Resource Allocation} (ORA) with replenishable budgets, in which a decision maker (agent) selects actions over a finite horizon of \( T \) rounds. We denote the action space of the agent by \( \mathcal{X} \subseteq \mathbb{R}^K \).  At each round $t \in [T]$,\footnote{Throughout this paper, we denote by $[a] := \{1,\ldots,a \}$ the set of the first $a \in \mathbb{N}$ natural numbers.} the agent observes an input tuple $\gamma_t = (f_t, \bm{b}_t)$, where 
$f_t: \mathcal{X} \to [0,1]$ is the reward function
and $\bm{b}_t = (b_{t,1},\ldots,b_{t,m}): \mathcal{X} \to [-1,1]^m$ is the vector of resource consumption functions, with $b_{t,i}: \mathcal{X} \to [-1,1]$ denoting the consumption function for resource $i \in [m]$.
The input functions are fully revealed to the agent before any decision is made. We assume that the maxima below are attained and use a fixed measurable rule to select a maximizer. After observing the input, the agent selects an action $x_t \in \mathcal{X}$, which yields a reward $f_t(x_t)$ and a consumption $b_{t,i}(x_t)$ for each resource $i \in [m]$. 
Unlike traditional ORA models, this setting allows for \emph{non-monotonic} resource utilization: a negative consumption, $b_{t,i}(x_t) < 0$, means that action $x_t$ \emph{replenishes} the budget of resource $i$ rather than consuming it.
The input tuples may be either \emph{stochastic}, \emph{i.e.}, generated i.i.d.\ across rounds according to a fixed probability distribution, or \emph{adversarial}, \emph{i.e.}, chosen arbitrarily by an adversary. Throughout the paper, we denote by $\Gamma$ the set of all the possible input tuples $\gamma = (f, \bm{b})$.

The decision maker maintains an $m$-dimensional vector $\bm{B}_t \in \mathbb{R}_{\geq 0}^m$ representing the available budget for each resource, which is updated at each round.
The initial budget is $T \rho$ for each resource, where $\rho \in (0,1)$ denotes the per-round budget; equivalently, we let $\bm{\rho} = \rho \bm{1} \in \mathbb{R}_{\geq 0}^m$ denote the vector of per-round budgets.
Following the standard assumption in the ORA literature (see, \emph{e.g.},~\citep{balseiro2023best}), we assume the existence of a \emph{void action} $\varnothing \in \mathcal{X}$. To allow recovery from budget depletion when $\beta>0$, we require the void action to provide a guaranteed replenishment of every resource. Formally, there exists a constant $\beta \in [0,1]$ such that, for every input tuple $(f,\bm{b}) \in \Gamma$, we have $f(\varnothing) = 0$ and $b_i(\varnothing) \leq -\beta$ for every $i \in [m]$.

\begin{remark}[On the void action and replenishment]
    We assume that the void action $\varnothing$ provides a replenishment of at least $\beta$ for every input tuple, even in the stochastic setting. This assumption captures realistic ORA applications in which a fallback option, such as pausing operations or acquiring additional resources, guarantees a minimum replenishment regardless of the stochasticity of incoming input requests. The lower bound $\beta$ need not be known to the algorithm. The case $\beta = 0$ includes standard monotonic ORA.
\end{remark}

Given a specific sequence of inputs $\bm{\gamma} = (\gamma_1, ..., \gamma_T),$ where $\gamma_t = (f_t, \bm{b}_t)$ for all $t \in [T]$, we define the \emph{cumulative reward} of the agent over the time horizon $T$ as: 
\[\texttt{Rew}(T, \bm{\gamma}):= \sum_{t \in [T]} f_t(x_t).\]

The goal of the agent is to maximize $\texttt{Rew}(T, \bm{\gamma})$, while satisfying the resource consumption constraint, \emph{i.e.}: 
\[\sum_{t \in [T]}  b_{t,i}(x_t) \leq \rho T \quad \forall i \in [m].\]

\subsection{Lagrangian Formulation} 
A crucial tool to handle resource constraints across the time horizon is the Lagrangian formulation. This approach allows us to incorporate the budget constraints directly into the objective function. For any round $t \in [T]$, given the instantaneous reward function $f_t$, the resource consumption vector $\bm{b}_t$ and the per-iteration budget $\rho$, we define the instantaneous Lagrangian function as: 
\[\mathcal{L}_{f_t,\bm{b}_t}(x, \bm{\lambda}) := f_t(x) - \sum_{i=1}^{m} \lambda_i (b_{t,i}(x) - \rho) = f_t(x) - 
\langle \bm{\lambda}, \bm{b}_t(x) - \bm{\rho} \rangle,\] 
where $\bm{\lambda} \in \mathbb{R}^m_{\geq 0}$ is the vector of dual multipliers.

In this formulation, the dual variables $\lambda_i$ act as penalty weights for resource utilization. For actions consuming more than $\rho$ units of resource $i$, the term $b_{t,i}(x) - \rho$ is positive, thus penalizing the Lagrangian objective. However, in our framework with replenishable budgets, playing an action that restores a resource $j$ guarantees a strictly negative term $b_{t,j}(x) - \rho$ to the objective. This structural property naturally encourages the algorithm to select replenishing actions when the dual multipliers become large, preventing the budget to be prematurely exhausted. 


\subsection{Stochastic Setting} \label{st}
In the stochastic setting, we assume that the input tuples $\gamma_t = (f_t, b_t)$ are sampled i.i.d.\ from a fixed but unknown probability distribution $\mathcal{P}$ over the input space $\Gamma$. 
Formally, the input sequence is generated as $\bm{\gamma} = (\gamma_1,..., \gamma_T) \sim \mathcal{P}^T$. 
The standard baseline for this setting is defined as the optimal strategy that, in expectation over all input realizations, maximizes rewards while satisfying constraints, \emph{i.e.}, $\text{OPT}_{\text{stoc}} = \mathbb{E}_{\bm{\gamma}\sim \mathcal{P}^T}[\text{OPT}(\bm{\gamma})]$, where: 
\[
	\text{OPT}(\bm{\gamma}):=\begin{cases}
		\max_{(x_1,...,x_T) \in \mathcal{X}^T} &  \sum_{t \in [T]} f_t(x_t)\\
		\,\,\, \textnormal{s.t.} & \sum_{t \in [T]}  b_{t,i}(x_t) \leq \rho T \quad \forall i \in [m] \nonumber.
	\end{cases}
\]
In the stochastic setting, we evaluate the performance of the algorithm by using the \emph{cumulative regret} over $T$ rounds, defined as: 
 \[R_T := \text{OPT}_\text{stoc} - \mathbb{E}_{\bm{\gamma} \sim \mathcal{P}^T}[\texttt{Rew}(T, \bm{\gamma})].\] 

\subsection{Adversarial Setting} \label{ad}
In the adversarial setting, we assume that the input tuple $\gamma_t$ is chosen arbitrarily at each round $t$ by an adversary. Exploiting the replenishable nature of the resources, we can evaluate the agent's performance against a more challenging benchmark, the dynamic \emph{unconstrained} optimum. Given a fixed input sequence $\bm{\gamma} = (\gamma_1, \ldots, \gamma_T) \in \Gamma^T$, we define $\text{OPT}_\text{adv}$ as: 
\[
    \text{OPT}_\text{adv} \coloneqq \max_{(x_1,...,x_T) \in \mathcal{X}^T} \sum_{t \in [T]} f_t(x_t).
\]
Due to the fundamental impossibility of simultaneously achieving sublinear regret and constraint satisfaction under fully adversarial inputs (see, \emph{e.g.},~\citep{balseiro2023best}), the standard regret definition is not viable for this setting. To overcome this issue, we measure the performance of our algorithm against a fraction $\alpha \in (0,1)$ of the optimum, called $\alpha\text{-}$regret: 
\[
    \alpha\text{-}R_T := \alpha \cdot\text{OPT}_\text{adv} - \texttt{Rew}(T),
\]  
where $\alpha := \frac{\rho + \beta}{1 + \beta}$ and we let $\texttt{Rew}(T)\coloneqq \texttt{Rew}(T,\bm{\gamma})$, as the dependence on the deterministic sequence of input tuples $\bm{\gamma}$ is clear from the context. 

Notice that the replenishable nature of the environment---specifically, the existence of the void action guaranteeing a replenishment $\beta > 0$---implicitly provides a Slater's condition with margin $\rho + \beta$. If $\beta = 0$, no strictly positive replenishment is guaranteed, and the Slater's margin becomes the known per-round budget as in classical budget-constrained online allocation (see, \emph{e.g.},~\citep{balseiro2023best}). 

\section{Algorithm} 
In this section, we describe the behavior of \textsc{Replenishing Dual Gradient Descent}, which is the main algorithmic contribution of this work.

The pseudocode of the algorithm is provided in Algorithm~\ref{alg:alg1}. It receives as input the time horizon $T$, the learning rate $\eta$ and the per-round budget $\rho$. As a first step, it initializes the dual variables $\bm{\lambda}_1$ to zero and instantiates the starting available budget $\bm{B}_1 \coloneqq T \bm{\rho}$ (Line~\ref{algLine:2}).

\begin{algorithm}[t]
	\caption{\textsc{Replenishing Dual Gradient Descent}}
	\label{alg:alg1}
	\begin{algorithmic}[1]
		\Require $T$, $\eta$, $\rho$  \label{algLine:0}
        \State $\bm{\lambda}_1\gets \bm{0}$ , $\bm{B}_1 \gets T \bm{\rho}$ \label{algLine:2}
		\For{$t = 1, 2, \ldots , T$} 
        \State Observe input tuple $\gamma_t = (f_t, \bm{b}_t)$ \label{algLine:3}
        \If {$\exists i \in [m]: B_{t,i} < 1$} \label{algLine:4} 
        \State $x_t \gets \varnothing$ \label{algLine:5} 
        \State $\bm{B}_{t+1} \gets \bm{B}_{t} - \bm{b}_{t}(\varnothing)$ \Comment{Overall budget update} \label{algLine:6} 
        \State $\bm{\lambda}_{t+1} \gets \bm{\lambda}_{t}$ \label{algLine:7}
        \Else 
        \State $x_t \gets \argmax_{x \in \mathcal{X}}\mathcal{L}_{f_t,\bm{b}_t}(x, \bm{\lambda}_t)$ \label{algLine:8} 
        \State $\bm{B}_{t+1} \gets \bm{B}_{t} - \bm{b}_{t}(x_t) $ \Comment{Overall budget update} \label{algLine:9} 
        \State $\bm{\lambda}_{t+1} \gets \Pi_{\mathbb{R}^m_{\geq 0}}(\bm{\lambda}_t + \eta (\bm{b}_t(x_t) - \bm{\rho}))$ \Comment{Dual update} \label{algLine:10}
        \EndIf
		\EndFor
	\end{algorithmic}
\end{algorithm}

At each round $t$, after observing the current input tuple $\gamma_t$ (Line~\ref{algLine:3}), the algorithm verifies that the available budget is sufficient to proceed safely (Line~\ref{algLine:4}). Since the maximum resource consumption for any action is bounded by 1 (Section~\ref{sec:preliminaries}), a remaining budget $B_{t,i} < 1$ implies that selecting a budget-consuming action could lead to a hard constraint violation for resource $i$. Therefore, if this critical threshold is reached, the algorithm acts conservatively and plays the void action $\varnothing$ (Line~\ref{algLine:5}). This action does not consume any budget; instead, it provides a guaranteed replenishment ($\beta \geq 0$) effectively increasing the remaining capacity (Line~\ref{algLine:6}).
Crucially, during the rounds in which the algorithm is forced to play the void action, the dual variables are not updated (Line~\ref{algLine:7}). Pausing the dual updates is a fundamental design choice: if the update were performed during replenishment, the dual multipliers would be artificially decreased, causing the algorithm to over-consume resources once the budget is recovered. 

If the available budget is safe ($B_{t,i} \geq 1$ for all $i$), the algorithm selects an action $x_t$ by greedily maximizing the instantaneous Lagrangian function (Line~\ref{algLine:8}). The available overall budget is updated based on the consumption given by the selected action (Line~\ref{algLine:9}). As standard in ORA framework, this primal update requires solving a single optimization problem over the action space $\mathcal{X}$, avoiding the computational burden of periodically re-solving large-scale linear programs over the time horizon.

Notice that the budget remains non-negative at every round. Indeed, an active round starts with $B_{t,i}\geq1$ and consumes at most 1 unit of each resource, while a forced round only replenishes resources. Thus, for every $s\in[T]$ and $i\in[m]$, we have $\sum_{t=1}^s b_{t,i}(x_t)\leq\rho T$.

Finally, the algorithm updates the dual variables $\bm{\lambda}_t$ through a projected \emph{Online Gradient Descent} (OGD) step onto the positive orthant (Line~\ref{algLine:10}). Intuitively, if the selected action consumed more budget than the per-round budget $\bm{\rho}$, the gradient step increases the multipliers, effectively penalizing ``expensive'' actions in subsequent rounds. Conversely, if resources are underutilized, the multipliers decrease, encouraging the algorithm to exploit higher-reward, budget-consuming actions in the future.

\section{Theoretical results}\label{th}
In this section, we prove the theoretical guarantees attained by Algorithm~\ref{alg:alg1}. In particular, we first show that the Lagrangian multipliers remain bounded during the learning dynamics. Subsequently, we analyze the cumulative regret in the stochastic setting and the cumulative $\alpha$-regret in the adversarial setting. 

\subsection{Bounding the Lagrangian Multipliers} \label{lambda} 
A key result in our analysis is showing how the dual multipliers are automatically bounded by properties of the dual regret minimizer. In standard ORA frameworks (see, \emph{e.g.}, \cite{balseiro2023best}) resource consumption is assumed to be strictly non-negative. This monotonic depletion of resources ensures that the dual variables naturally self-regulate and remain bounded without requiring explicit projection onto a bounded set. Moreover, in these settings, the strict feasibility of the constraints is determined solely by the per-round budget $\rho$, typically known by the algorithm. In our setting, the non-monotonicity of the budget dynamics could, in principle, destabilize the dual space by making different components of the multiplier vector increase or decrease in a non-aligned way. Thus, we need to formally prove that the primal-dual learning dynamics still naturally self-regulate despite the presence of negative components. To do so, we exploit the existence of the guaranteed replenishment margin $\beta > 0$, which effectively increases the implicit Slater's condition to $\rho + \beta$. This augmented margin allows the algorithm to maintain a tightly bounded dual space.  
\begin{restatable}{lemma}{selfb}\label{selfb}
    In both the stochastic and the adversarial setting, Algorithm~\ref{alg:alg1} with $0 < \eta \leq \nicefrac{1}{2m}$ guarantees that
    the $\ell_1$-norm of the dual variables is bounded for all $t \in [T+1]$ as:
    \[ \|\bm{\lambda}_t\|_1 \leq \frac{4\sqrt{m}}{\rho+\beta}.\] 
\end{restatable} 
This result formally captures the structural power of our algorithm. Indeed, the maximum magnitude of the multipliers is inversely proportional to the effective margin $\rho + \beta$. We prove Lemma~\ref{selfb} by controlling the squared Euclidean norm of the dual variables. During active rounds, the greedy primal step dominates the void action. Thus, we obtain:
\[\mathcal{L}_{f_t,\bm{b}_t}(x_t, \bm{\lambda}_t) \geq \mathcal{L}_{f_t,\bm{b}_t}(\varnothing, \bm{\lambda}_t) \geq (\rho + \beta) \|\bm{\lambda}_t\|_1.\]
Since rewards are bounded by 1, this implies:
\[\langle \bm{\lambda}_t,\bm{b}_t(x_t)-\bm{\rho}\rangle \leq 1-(\rho+\beta)\|\bm{\lambda}_t\|_1.\]
Moreover, the projected OGD update satisfies:
\[\|\bm{\lambda}_{t+1}\|_2^2 \leq \|\bm{\lambda}_t\|_2^2 + 2\eta\bigl(1-(\rho+\beta)\|\bm{\lambda}_t\|_1\bigr)+4\eta^2m.\]
Then, for $\eta \leq \nicefrac{1}{2m}$, the squared norm cannot increase whenever $\|\bm{\lambda}_t\|_1 \geq \nicefrac{2}{\rho+\beta}$. Below this threshold, a single update has bounded size. During forced rounds, the dual variables remain unchanged. Combining these properties yields the stated bound. The full details of the proof can be found in Appendix~\ref{app}.

Notice that the learning rates in Theorems~\ref{stre} and~\ref{advr} may exceed $\nicefrac{1}{2m}$ for small horizons. In this case, the stated regret bounds already exceed $T$ and follow from bounded rewards. Thus, we only apply Lemma~\ref{selfb} when its learning-rate condition holds.

\subsection{Stochastic Setting} 
In this section, we analyze the performance of Algorithm~\ref{alg:alg1} in the stochastic setting, where inputs are drawn from a fixed but unknown distribution. 
\begin{restatable}{theorem}{stre} \label{stre}
    In the stochastic setting, assuming the replenishment parameter $\beta$ is known, Algorithm~\ref{alg:alg1} with learning rate $\eta \coloneqq \frac{3}{2 (\rho + \beta)} \sqrt{\frac{1}{T}}$ guarantees: 
    \[R_T \leq \frac{2 + 12 m \sqrt{T}}{\rho + \beta}.\] 
    If $\beta$ is unknown, Algorithm~\ref{alg:alg1} with learning rate $\eta \coloneqq \frac{3}{2\rho} \sqrt{\frac{1}{T}}$ guarantees: 
    \[R_T \leq \frac{2}{\rho + \beta} + 6m\left( \frac{\rho}{(\rho + \beta)^2} + \frac{1}{\rho} \right) \sqrt{T}.\]
\end{restatable}

Theorem~\ref{stre} establishes that Algorithm~\ref{alg:alg1} attains $\mathcal{O}(\sqrt{T})$ cumulative regret bound, which is the optimal guarantee for stochastic ORA. The proof includes two phases: we first derive an upper bound on the stochastic benchmark $\text{OPT}_{\text{stoc}}$ by invoking weak duality and defining the expected Lagrangian over the input distribution:
\[\bar{\mathcal{L}}_{f,\bm{b}}(\bm{\lambda} | \mathcal{P}) := \mathbb{E}_{(f, \bm{b}) \sim \mathcal{P}} [\max_{x \in \mathcal{X}}\mathcal{L}_{f,\bm{b}}(x, \bm{\lambda})].\] 
Next, we evaluate algorithmic performance across the active decision rounds. By taking the conditional expectation of the instantaneous reward $f_t(x_t)$ with respect to the filtration of past events $\Gamma_{t-1} = \{\gamma_1,..., \gamma_{t-1} \}$, we can relate it to the instantaneous Lagrangian via a martingale argument:
\[\mathbb{E} [f_t(x_t) \cdot \mathbb{I} \{t \in \mathcal{T}_{\text{active}}\} | \Gamma_{t-1}] = \mathbb{I} \{t \in \mathcal{T}_{\text{active}}\}  \bar{\mathcal{L}}_{f,\bm{b}}(\bm{\lambda}_t | \mathcal{P}) +  \mathbb{E}  [ \mathbb{I} \{t \in \mathcal{T}_{\text{active}}\} \langle \bm{\lambda}_t, \bm{b}_t(x_t) - \bm{\rho} \rangle | \Gamma_{t-1}] .\]
By summing over the time horizon and applying the Tower Rule, this step allows us to upper bound the overall regret by the expected number of forced replenishments and the cumulative dual penalty during the active rounds $\mathcal{T}_{\text{active}}$: 
\[R_T \leq \mathbb{E}_{\bm{\gamma}} [ |\mathcal{T}_{\text{forced}}| ] - \mathbb{E}\left[ \sum_{t \in \mathcal{T}_{\text{active}}} \langle \bm{\lambda}_t, \bm{b}_t(x_t) - \bm{\rho} \rangle \right]. \]
Finally, we bound the cumulative dual penalty by exploiting the no-regret guarantees of the OGD update. Crucially, the physical limits of the budget ensure that whenever the safety threshold is reached---thus triggering a forced replenishment---the overall resource consumption provides a strict lower bound on the dual penalty. The complete proof can be found in Appendix~\ref{app}. 

\begin{remark}
    As the theorem shows, the algorithm does not need to know the true replenishment parameter $\beta$ to guarantee the optimal $\widetilde{\mathcal{O}} (\sqrt{T})$ rate. When $\beta$ is unknown, the learning rate is tuned solely on the known per-round budget $\rho$, since $\rho \leq \rho + \beta$. While the constant factors inevitably worsen, suggesting that an accurate estimate of $\beta$ is useful to tighten the bound, the sublinearity is preserved also in this case.
\end{remark}

\subsection{Adversarial Setting}
Finally, in this section, we analyze the performance of the algorithm in the adversarial regime, where inputs may arbitrarily change at each round.
\begin{restatable}{theorem}{advr} \label{advr}
    In the adversarial setting, assuming the replenishment parameter $\beta$ is known, Algorithm~\ref{alg:alg1} with learning rate $\eta \coloneqq \frac{3}{2 (\rho + \beta)} \sqrt{\frac{1}{T}}$ guarantees: 
    \[\alpha\text{-}R_T \leq \frac{2 + 12 m \sqrt{T}}{\rho + \beta}.\] 
    If $\beta$ is unknown, Algorithm~\ref{alg:alg1} with learning rate $\eta \coloneqq \frac{3}{2 \rho} \sqrt{\frac{1}{T}}$ guarantees: 
    \[\alpha\text{-}R_T \leq \frac{2}{\rho + \beta} + 6m\left( \frac{\rho}{(\rho + \beta)^2} + \frac{1}{\rho} \right) \sqrt{T}.\]
\end{restatable}

Theorem~\ref{advr} proves that \textsc{Replenishing Dual Gradient Descent} ensures an $\mathcal{O}(\sqrt{T})$ bound for the $\alpha$-regret under fully adversarial inputs. We establish this result by leveraging the regret guarantees of the unconstrained OGD updates. The analysis naturally partitions the time horizon into the set of active decision rounds $\mathcal{T}_{\text{active}}$ and the set of forced replenishment rounds $\mathcal{T}_{\text{forced}}$. During the forced replenishments, the algorithm conservatively plays the void action $\varnothing$, thus the regret incurred in this phase is trivially bounded by $|\mathcal{T}_{\text{forced}}|$.

To analyze the active rounds, we rely on the fact that the primal update greedily maximizes the instantaneous Lagrangian. We benchmark the algorithm's performance against a tailored, randomized reference strategy $\xi_{t}^{\diamond}$. Specifically, at each step $t$, this strategy selects the hindsight optimal action $x^*_t$ with probability $\alpha \coloneqq \frac{\rho + \beta}{1 + \beta}$ and defaults to the void action with probability $1 - \alpha$. The parameter $\alpha$ is determined by the requirement to offset the worst-case resource consumption of the offline optimal sequence with the guaranteed margin $\beta$ provided by the void action. 

Since $f_t(\varnothing) = 0$, the expected reward of this strategy scales exactly by $\alpha$: 
\[\mathbb{E}_{x\sim \xi_t^\diamond}[f_t(x)] = \alpha \cdot f_t(x_t^*) + (1 - \alpha) f_t(\varnothing) \geq \alpha \cdot f_t(x_t^*).\]
Simultaneously, thanks to the guaranteed replenishment margin $\beta$ provided by the void action, this randomized policy remains strictly feasible in expectation: 
\[\mathbb{E}_{x\sim \xi_t^\diamond}[b_{t,i}(x)] \leq \alpha \cdot 1 + (1 - \alpha)(-\beta) = \rho.\]
This allows us to reduce the $\alpha$-regret analysis to bounding the cumulative dual penalty during the active rounds $\mathcal{T}_{\text{active}}$, plus the trivial regret incurred during forced replenishments $\mathcal{T}_{\text{forced}}$: 
\[\alpha\text{-}R_T \leq  - \sum_{t \in \mathcal{T}_{\text{active}}} \langle \bm{\lambda}_t, \bm{b}_t(x_t) - \bm{\rho} \rangle + |\mathcal{T}_{\text{forced}}|.\]
Similarly to the stochastic analysis, we control the penalty via the no-regret properties of the OGD updates. The complete proof is provided in Appendix~\ref{app}.

\section{Discussion and Open Problems}
In this work, we introduce a generalized model for online resource allocation that includes replenishable budgets, departing from the standard monotonic consumption assumption. We present \textsc{Replenishing Dual Gradient Descent}, a dual-based algorithm that achieves optimal $\widetilde{\mathcal O}(\sqrt{T})$ expected regret in stochastic environments and $\widetilde{\mathcal O}(\sqrt{T})$ $\alpha$-regret in adversarial settings, while strictly satisfying the budget constraints. Our analysis shows that the presence of an active replenishment margin ($\beta > 0$) tightens the regret bounds and improves the adversarial competitive ratio $\alpha$, without requiring prior knowledge of the parameter $\beta$.

This problem formulation opens several directions for future theoretical research. For instance: 

\begin{itemize}
    \item Replenishment in expectation: we currently assume that the void action $\varnothing$ provides a replenishment of at least $\beta$ for every input tuple. A natural relaxation is to require this lower bound only in conditional expectation. In this context, it would be interesting to study whether the algorithm can estimate $\beta$, rather than relying on the worst case bound when $\beta$ is unknown.
    \item Resource-specific replenishing actions: our current model assumes a single void action that uniformly replenishes all $m$ resources by at least $\beta$. In many real-world applications (e.g., inventory systems), however, replenishment is inherently resource-specific. It may be available through multiple independent replenishing actions with different margins, instead of a single global one. Extending our dual-based framework to this more general setting would broaden its practical applicability.
\end{itemize}

\bibliographystyle{plainnat}
\bibliography{example_paper}

\newpage
\appendix
\section{Additional Related Works} \label{rw}
A line of work related to our setting studies decision making under more general long-term constraints, which unlike budget constraints are not hard, i.e., they can suffer temporary violations over the time horizon. In this framework, the quantity to control during the learning process is the cumulative constraint violation, which can receive negative per-round contribution: an action may satisfy a constraint with slack, compensating for violations incurred in earlier rounds. Despite showing this ``compensating'' feature, this setting fundamentally differs from the replenishable ORA setting. For general constraints, no single action is guaranteed a priori to offset violations across all constraints simultaneously. In our setting, a single, known action $\varnothing$ provides a contribution of at most $-\beta$ for every input tuple. 
\citep{mahdavi2012tradingregretefficiencyonline} and \citep{jenatton2015adaptivealgorithmsonlineconvex} were among the first to show that sublinear regret and constraint violation can be attained simultaneously under fixed constraints. \citep{Mannor}, \citep{yu2017onlineconvexoptimizationstochastic} and \citep{sun2022nearoptimalprimaldualalgorithmsquantitybased} later extended this guarantee for time-varying constraints. A closely related model, \emph{Bandits with Knapsacks} (BwK) was introduced by \citep{Badanidiyuru_2013} for budget constraints and was sharpened for the stochastic setting by several works, such as \citep{agrawal_globalconvex}. The results were extended to adversarial and non-stationary environments by \citep{immorlica2023adversarialbanditsknapsacks}, \citep{kesselheim2020} and \citep{castiglioni2022online}. Most relevant to our discussion, \citep{bernasconi2024noregret} and \citep{castiglioni2024online} generalize BwK beyond budgets accommodating arbitrary long-term constraints.

\section{Omitted Proofs and Lemmas} \label{app}
\subsection{Omitted Proofs and Lemmas of Section \ref{th}}
\begin{lemma} \label{ogd}
    Let $\mathcal{T}_{\text{forced}}$ be the set of rounds in which the budget condition forces Algorithm~\ref{alg:alg1} to play $\varnothing$, and let $\mathcal{T}_{\text{active}}=[T]\setminus\mathcal{T}_{\text{forced}}$. Let the dual algorithm be OGD on $\mathbb{R}^m_{\geq 0}$ with learning rate $\eta>0$, with no updates during forced rounds. Then, for any $\bm{\mu}\in\mathbb{R}^m_{\geq 0}$ and $[t_1,t_2]\subseteq[T]$, it holds:
    \begin{equation*}
    \sum_{t \in [t_1,t_2]\cap\mathcal{T}_{\text{active}}}\langle\bm{\lambda}_t,\bm{b}_t-\bm{\rho}\rangle
    \geq \sum_{t \in [t_1,t_2]\cap\mathcal{T}_{\text{active}}}\langle\bm{\mu},\bm{b}_t-\bm{\rho}\rangle-\frac{\|\bm{\lambda}_{t_1}-\bm{\mu}\|_2^2}{2\eta}-2\eta Tm,
    \end{equation*}
    where $\bm{b}_t=\bm{b}_t(x_t)$. 
\end{lemma}
\begin{proof}
    For each active round, by the properties of Euclidean projection onto a convex set,
    \begin{align*}
        \|\bm{\lambda}_{t+1} - \bm{\mu}\|_2^2 &= \|\Pi_{\mathbb{R}^m_{\geq 0}} (\bm{\lambda}_t + \eta (\bm{b}_t - \bm{\rho}) ) - \bm{\mu} \|_2^2 \\ 
        &\leq \|\bm{\lambda}_t + \eta (\bm{b}_t - \bm{\rho})  - \bm{\mu}\|_2^2 \\
        &= \| \bm{\lambda}_t - \bm{\mu} \|_2^2 + \eta^2 \|\bm{b}_t - \bm{\rho} \|_2^2 + 2\eta \langle \bm{\lambda}_t - \bm{\mu}, \bm{b}_t - \bm{\rho}  \rangle. 
    \end{align*} 
    Rearranging the terms: 
    \begin{align*}
        2\eta \langle \bm{\mu} - \bm{\lambda}_t, \bm{b}_t - \bm{\rho}  \rangle &\leq \| \bm{\lambda}_t - \bm{\mu} \|_2^2 - \|\bm{\lambda}_{t+1} - \bm{\mu}\|_2^2 + \eta^2 \|\bm{b}_t - \bm{\rho} \|_2^2 \\ 
        \langle \bm{\mu} - \bm{\lambda}_t, \bm{b}_t - \bm{\rho}  \rangle &\leq \frac{1}{2\eta} \| \bm{\lambda}_t - \bm{\mu} \|_2^2 - \frac{1}{2\eta}\|\bm{\lambda}_{t+1} - \bm{\mu}\|_2^2 + \frac{\eta}{2} \|\bm{b}_t - \bm{\rho} \|_2^2.
    \end{align*}
    Summing over active rounds, the squared-distance terms telescope since the dual variables remain unchanged during forced rounds. Thus, we obtain:
    \begin{align*}
        \sum_{t \in [t_1,t_2]\cap\mathcal{T}_{\text{active}}}\langle \bm{\mu} - \bm{\lambda}_t, \bm{b}_t - \bm{\rho}  \rangle &\leq \frac{1}{2\eta} \| \bm{\lambda}_{t_1} - \bm{\mu} \|_2^2 - \frac{1}{2\eta}\|\bm{\lambda}_{t_2 + 1} - \bm{\mu}\|_2^2 + \frac{\eta}{2} \sum_{t \in [t_1,t_2]\cap\mathcal{T}_{\text{active}}} \|\bm{b}_t - \bm{\rho} \|_2^2 \\ 
        &\leq \frac{1}{2\eta} \| \bm{\lambda}_{t_1} - \bm{\mu} \|_2^2 - \frac{1}{2\eta}\|\bm{\lambda}_{t_2 + 1} - \bm{\mu}\|_2^2 + \frac{\eta}{2} T m (1 + \rho)^2 \\ 
        &\leq \frac{1}{2\eta} \| \bm{\lambda}_{t_1} - \bm{\mu} \|_2^2 + 2\eta T m, 
    \end{align*}
    where in the last step we used the fact that $\rho < 1 $, yielding $(1 + \rho)^2 \leq 4$.
    Rearranging the terms gives the stated bound.
\end{proof} 

\begin{lemma} \label{aux1}
    Let $\mathcal{T}_{\text{forced}} \subseteq [T]$ be the set of rounds in which the algorithm is forced to play the void action---due to the triggering of the budget condition---and $\mathcal{T}_{\text{active}} = [T] \setminus \mathcal{T}_{\text{forced}}$ be the set of active rounds, in which the algorithm maximizes the Lagrangian. Assume $\mathcal{T}_{\text{forced}}\neq\varnothing$. Let $t_{\varnothing} \in \mathcal{T}_{\text{forced}}$ be the last forced round, and $j \in [m]$ the scarcest resource in round $t_{\varnothing}$. Formally: 
    \[
        t_{\varnothing} \coloneqq \max \mathcal{T}_{\text{forced}}, \qquad
        j \coloneqq \argmin_{i \in [m]} B_{t_{\varnothing}, i}. 
    \]
    Then the cumulative consumption of resource $j$ during the active rounds preceding $t_{\varnothing}$ can be bounded as:  
    \[
        \sum_{t \in \mathcal{T}_{\text{active}} : t < t_{\varnothing}} b_{t,j}(x_t) \geq \rho T - 2 + \beta |\mathcal{T}_{\text{forced}}|.
    \]
\end{lemma}
\begin{proof}
From the budget update rule, for any $t \in [T]$ we have $
    B_{t+1, j} = B_{t,j} - b_{t,j}(x_t). 
$
Rearranging and summing up to round $t_{\varnothing} - 1$, we get: 
\begin{align}
    \sum_{t \in [t_{\varnothing} - 1]} b_{t,j}(x_t) &= \sum_{t \in [t_{\varnothing} - 1]} (B_{t, j} - B_{t+1,j}) \nonumber \\ 
    &= B_{1,j} - B_{t_{\varnothing}, j} \nonumber \\  
    &> \rho T - 1, \label{ex}
\end{align}
where Inequality~\eqref{ex} follows from the initial budget $B_{1,j} = \rho T$ and the fact that $B_{t_{\varnothing}, j} < 1$ since round $t_{\varnothing}$ triggered a replenishment. 

Exploiting this result, we can write: 
\begin{align}
    \sum_{t \in \mathcal{T}_{\text{active}} : t < t_{\varnothing}} b_{t,j}(x_t) &= \sum_{t \in [t_{\varnothing} - 1]} b_{t,j}(x_t) - \sum_{t \in \mathcal{T}_{\text{forced}} \setminus \{t_{\varnothing}\}} b_{t,j}(x_t) \nonumber \\
    &\geq \rho T - 1  - \sum_{t \in \mathcal{T}_{\text{forced}} \setminus \{t_{\varnothing}\} } b_{t,j}(\varnothing) \label{pv} \\ 
    &\geq \rho T - 1 + \beta (|\mathcal{T}_{\text{forced}}| - 1) \label{var} \\ 
    &\geq \rho T - 2 + \beta |\mathcal{T}_{\text{forced}}|, \nonumber
\end{align}
where Inequality~\eqref{pv} follows from Inequality~\eqref{ex} and Inequality~\eqref{var} holds by definition of the void action $\varnothing$. The last step is due to the assumption $\beta \leq 1$. 
This concludes the proof.
\end{proof}

\begin{lemma} \label{b4}
    For each $t\in\mathcal{T}_{\text{active}}$, Algorithm~\ref{alg:alg1} with learning rate $\eta>0$ guarantees:
    \[\|\bm{\lambda}_{t+1}\|_2^2 \leq \|\bm{\lambda}_t\|_2^2 + 2\eta\bigl(1-(\rho+\beta)\|\bm{\lambda}_t\|_1\bigr)+4\eta^2m.\]
\end{lemma}
\begin{proof}
    Since $x_t$ maximizes the Lagrangian, we have,
    \begin{align*}
        f_t(x_t)-\langle\bm{\lambda}_t,\bm{b}_t(x_t)-\bm{\rho}\rangle
        &\geq -\langle\bm{\lambda}_t,\bm{b}_t(\varnothing)-\bm{\rho}\rangle \\
        &\geq (\rho+\beta)\|\bm{\lambda}_t\|_1.
    \end{align*}
    Thus, since $f_t(x_t)\leq1$, we obtain:
    \[\langle\bm{\lambda}_t,\bm{b}_t(x_t)-\bm{\rho}\rangle\leq1-(\rho+\beta)\|\bm{\lambda}_t\|_1.\]
    By non-expansiveness of the projection and $\|\bm{b}_t(x_t)-\bm{\rho}\|_2^2\leq4m$, we have:
    \begin{align*}
        \|\bm{\lambda}_{t+1}\|_2^2
        &\leq\|\bm{\lambda}_t\|_2^2+2\eta\langle\bm{\lambda}_t,\bm{b}_t(x_t)-\bm{\rho}\rangle
        +\eta^2\|\bm{b}_t(x_t)-\bm{\rho}\|_2^2 \\
        &\leq\|\bm{\lambda}_t\|_2^2+2\eta\bigl(1-(\rho+\beta)\|\bm{\lambda}_t\|_1\bigr)+4\eta^2m.
    \end{align*}
    This concludes the proof.
\end{proof}

\begin{lemma} \label{aux2} 
    Let the dual variables be generated by Algorithm~\ref{alg:alg1} with $0<\eta\leq\nicefrac{1}{2m}$.
    It holds 
    \[
    \sum_{t \in \mathcal{T}_{\text{active}}} \langle \bm{\lambda}_t, \bm{b}_t(x_t) - \bm{\rho} \rangle \geq |\mathcal{T}_{\text{forced}}| - \frac{2}{\rho + \beta} - \frac{9 m}{\eta (\rho + \beta)^2} - 4 \eta T m.
    \]
\end{lemma}
\begin{proof}
If $\mathcal{T}_{\text{forced}}=\varnothing$, Lemma~\ref{ogd} with $\bm{\mu}=\bm{0}$ and $\bm{\lambda}_1=\bm{0}$ gives a lower bound of $-2\eta Tm$, which implies the claim. Thus, assume $\mathcal{T}_{\text{forced}}\neq\varnothing$. We provide a lower bound on the term $\sum_{t \in \mathcal{T}_{\text{active}}} \langle \bm{\lambda}_t, \bm{b}_t(x_t) - \bm{\rho} \rangle$ employing Lemma~\ref{ogd} with different comparators for the rounds before and after $t_{\varnothing}$. Since the dual multipliers are not updated during forced rounds, Lemma~\ref{ogd} applies to the active rounds on either side of $t_{\varnothing}$. If either set is empty, the corresponding OGD bound holds trivially. 
\begin{itemize}
    \item ($t < t_{\varnothing}$) Let $j$ be the resource identified in Lemma~\ref{aux1}. We choose the comparator $\bm{\mu}$ such that $\mu_{j} = \frac{1}{\rho + \beta}$ and $\mu_{i} = 0$ for each $i \neq j$. We obtain:  
    \begin{align}
        \sum_{t \in \mathcal{T}_{\text{active}}: t < t_{\varnothing}} \langle \bm{\lambda}_t, \bm{b}_t(x_t) - \bm{\rho} \rangle 
        &\geq \sum_{t \in \mathcal{T}_{\text{active}}: t < t_{\varnothing}} \mu_j (b_{t,j}(x_t) - \rho) \nonumber\\
        &\quad- \frac{1}{2 \eta} \| \bm{\lambda}_1 - \bm{\mu}\|_2^2 - 2 \eta T m \nonumber
        \\&\geq \frac{1}{\rho + \beta} \left( \sum_{t \in \mathcal{T}_{\text{active}}: t < t_{\varnothing}} b_{t,j}(x_t) - \sum_{t \in \mathcal{T}_{\text{active}}: t < t_{\varnothing}} \rho \right) + \nonumber \\ &\mkern210mu  - \frac{1}{2\eta (\rho + \beta)^2} - 2 \eta T m \nonumber \\
        &\geq \frac{1}{\rho + \beta} \left( \rho T - 2 + \beta |\mathcal{T}_{\text{forced}}| - \sum_{t \in \mathcal{T}_{\text{active}}: t < t_{\varnothing}} \rho \right) + \nonumber \\ &\mkern210mu  - \frac{1}{2\eta (\rho + \beta)^2} - 2 \eta T m, 
        \label{a1} \\
        \intertext{where Inequality~\eqref{a1} follows from Lemma~\ref{aux1}. Notice that the total horizon budget $\rho T$ can be decomposed across the disjoint sets as $\rho T = \rho |\mathcal{T}_{\text{forced}}| + \sum_{t \in \mathcal{T}_{\text{active}}: t < t_{\varnothing}} \rho + \sum_{t \in \mathcal{T}_{\text{active}}: t > t_{\varnothing}} \rho $. Substituting this expansion we obtain:} 
        &= \frac{1}{\rho + \beta} \left( \rho |\mathcal{T}_{\text{forced}}| - 2 + \beta |\mathcal{T}_{\text{forced}}| + \sum_{t \in \mathcal{T}_{\text{active}}: t > t_{\varnothing}} \rho \right) + \nonumber \\ &\mkern210mu  - \frac{1}{2\eta (\rho + \beta)^2} - 2 \eta T m \nonumber \\
        &= \frac{\sum_{t \in \mathcal{T}_{\text{active}}: t > t_{\varnothing}} \rho}{\rho + \beta} - \frac{2}{\rho + \beta} + |\mathcal{T}_{\text{forced}}| \nonumber\\
        &\quad- \frac{1}{2\eta (\rho + \beta)^2} - 2 \eta T m \nonumber \\
        &\geq - \frac{2}{\rho + \beta} + |\mathcal{T}_{\text{forced}}| - \frac{1}{2\eta (\rho + \beta)^2} - 2 \eta T m, \label{ge0}   \end{align}
    where Inequality~\eqref{ge0} holds because $\sum \rho \geq 0$.
    
    \item ($t > t_{\varnothing}$) Since $t_{\varnothing}$ is the last round in which a replenishment is required, in the phase following $t_{\varnothing}$ the budget condition never forces the void action again. We choose the comparator $\bm{\mu} = \bm{0}$ to obtain: 
    \begin{align}
        \sum_{t \in \mathcal{T}_{\text{active}}: t > t_{\varnothing}} \langle \bm{\lambda}_t, \bm{b}_t(x_t) - \bm{\rho} \rangle 
        &\geq 0 - \frac{1}{2 \eta} \| \bm{\lambda}_{t_{\varnothing} + 1}\|_2^2 - 2 \eta T m \nonumber \\
        &\geq - \frac{1}{2 \eta} \left( \frac{4 \sqrt{m}}{\rho + \beta} \right)^2  - 2 \eta T m \label{sb} \\ 
        &\geq - \frac{16 m}{2 \eta (\rho + \beta)^2}  - 2 \eta T m, 
    \end{align} 
    where Inequality~\eqref{sb} holds by Lemma~\ref{selfb}. 
\end{itemize}

Putting everything together: 
\begin{align}
    \sum_{t \in \mathcal{T}_{\text{active}}} \langle \bm{\lambda}_t, \bm{b}_t(x_t) - \bm{\rho} \rangle
    &= \sum_{t \in \mathcal{T}_{\text{active}}: t < t_{\varnothing}} \langle \bm{\lambda}_t, \bm{b}_t(x_t) - \bm{\rho} \rangle \nonumber\\
    &\quad+\sum_{t \in \mathcal{T}_{\text{active}}: t > t_{\varnothing}} \langle \bm{\lambda}_t, \bm{b}_t(x_t) - \bm{\rho} \rangle \nonumber\\
    &\geq -\left( \frac{2}{\rho + \beta} - |\mathcal{T}_{\text{forced}}| + \frac{1}{2\eta (\rho + \beta)^2} + 2 \eta T m \right) \nonumber\\
    &\quad-\left(\frac{16 m}{2 \eta (\rho + \beta)^2} + 2 \eta T m\right) \nonumber\\
    &\geq |\mathcal{T}_{\text{forced}}| - \frac{2}{\rho + \beta} - \frac{9 m}{\eta (\rho + \beta)^2} - 4 \eta T m, \nonumber
\end{align}
which concludes the proof.
\end{proof}

\begingroup
\renewcommand{\thetheorem}{\getrefnumber{selfb}}
\selfb*
\endgroup
\begin{proof}
    We prove by induction that $\|\bm{\lambda}_t\|_2\leq\nicefrac{4}{\rho+\beta}$ for all $t\in[T+1]$. The claim holds at $t=1$ since $\bm{\lambda}_1=\bm{0}$. During forced rounds, the dual variables remain unchanged.

    Consider an active round $t$. If $\|\bm{\lambda}_t\|_1\geq\nicefrac{2}{\rho+\beta}$, Lemma~\ref{b4} gives:
    \[\|\bm{\lambda}_{t+1}\|_2^2\leq\|\bm{\lambda}_t\|_2^2-2\eta+4\eta^2m\leq\|\bm{\lambda}_t\|_2^2,\]
    where the last inequality follows from $\eta\leq\nicefrac{1}{2m}$.
    Otherwise, by non-expansiveness of the projection, we obtain:
    \begin{align*}
        \|\bm{\lambda}_{t+1}\|_2
        &\leq\|\bm{\lambda}_t\|_2+\eta\|\bm{b}_t(x_t)-\bm{\rho}\|_2 \\
        &\leq\|\bm{\lambda}_t\|_1+2\eta\sqrt{m} \\
        &\leq\frac{2}{\rho+\beta}+\frac{1}{\sqrt{m}}
        \\&\leq\frac{4}{\rho+\beta},
    \end{align*}
    since $\rho+\beta\leq2$ and $m\geq1$. Thus, the induction is complete. Finally,
    \[\|\bm{\lambda}_t\|_1\leq\sqrt{m}\|\bm{\lambda}_t\|_2\leq\frac{4\sqrt{m}}{\rho+\beta}
    ,\]
    which concludes the proof.
\end{proof}

\stre*
\begin{proof}
    First, consider the case $\eta>\nicefrac{1}{2m}$. If $\beta$ is known, the prescribed learning rate implies $\sqrt{T}<\nicefrac{3 m}{\rho + \beta}$. Thus,
    \[R_T\leq T<\frac{3 m\sqrt{T}}{\rho+\beta}\leq \frac{2 + 12 m \sqrt{T}}{\rho + \beta}.\]
    If $\beta$ is unknown, we similarly have $\sqrt{T}<\nicefrac{3 m}{\rho}$ and
    \[R_T\leq T<\frac{3 m\sqrt{T}}{\rho}\leq \frac{2}{\rho + \beta} + 6m\left( \frac{\rho}{(\rho + \beta)^2} + \frac{1}{\rho} \right) \sqrt{T}.\]
    Therefore, both claims hold in this case. For the remainder of the proof, assume $0<\eta\leq\nicefrac{1}{2m}$.

    First, we provide an upper bound on the hindsight stochastic optimum. We define the expected Lagrangian evaluated with respect to the input distribution as:  
    \[\bar{\mathcal{L}}_{f,\bm{b}}(\bm{\lambda} | \mathcal{P}) := \mathbb{E}_{(f, \bm{b}) \sim \mathcal{P}} \left[\max_{x \in \mathcal{X}} \mathcal{L}_{f,\bm{b}}(x, \bm{\lambda}) \right].\]  
    By weak duality, for every fixed $\bm{\lambda}\in\mathbb{R}^m_{\geq0}$ and every input sequence $\bm{\gamma}$, we have:
    \[\text{OPT}(\bm{\gamma})\leq\sum_{t\in[T]}\max_{x\in\mathcal{X}}\mathcal{L}_{f_t,\bm{b}_t}(x,\bm{\lambda}).\]
    Taking expectations gives $\text{OPT}_{\text{stoc}}\leq T\bar{\mathcal{L}}_{f,\bm{b}}(\bm{\lambda}\mid\mathcal{P})$ for every $\bm{\lambda}\geq\bm{0}$. Notice that this is an inequality between deterministic functions of $\bm{\lambda}$. Thus, it also holds pointwise when evaluated at a random dual vector.
    Let $\bm{\bar{\lambda}}_{\mathcal{T}_{\text{active}}}:=\frac{1}{|\mathcal{T}_{\text{active}}|}\sum_{t\in\mathcal{T}_{\text{active}}}\bm{\lambda}_t$ when $\mathcal{T}_{\text{active}}\neq\varnothing$, and let $\bm{\bar{\lambda}}_{\mathcal{T}_{\text{active}}}=\bm{0}$ otherwise. Then,
    \begin{align}
        \text{OPT}_{\text{stoc}}
        &=\mathbb{E}_{\bm{\gamma}}\left[\frac{|\mathcal{T}_{\text{active}}|}{T}\text{OPT}_{\text{stoc}}\right]
        +\mathbb{E}_{\bm{\gamma}}\left[\frac{|\mathcal{T}_{\text{forced}}|}{T}\text{OPT}_{\text{stoc}}\right]\nonumber\\
        &\leq\mathbb{E}_{\bm{\gamma}}\left[|\mathcal{T}_{\text{active}}|\bar{\mathcal{L}}_{f,\bm{b}}(\bm{\bar{\lambda}}_{\mathcal{T}_{\text{active}}}\mid\mathcal{P})\right]
        +\mathbb{E}_{\bm{\gamma}}[|\mathcal{T}_{\text{forced}}|],\label{opt_bound}
    \end{align}
    where we used that rewards are in $[0,1]$. 
    
     During active rounds, Algorithm~\ref{alg:alg1} maximizes the instantaneous Lagrangian. Therefore, we can write the instantaneous reward as: 
    \[f_t(x_t) = \mathcal{L}_{f_t,\bm{b}_t}(x_t, \bm{\lambda}_t) + \langle \bm{\lambda}_t, \bm{b}_t(x_t) - \bm{\rho} \rangle.\] 

    Let $\Gamma_t = (\gamma_1, ..., \gamma_t)$ denote the history of observed inputs up to time $t$. Since $\bm{\lambda}_t$ and $\mathbb{I} \{t \in \mathcal{T}_{\text{active}}\}$ are $\Gamma_{t-1}$-measurable, taking the conditional expectation during active rounds we obtain: 
    \begin{align}
    \mathbb{E} [f_t(x_t)\cdot\mathbb{I}\{t\in\mathcal{T}_{\text{active}}\}\mid\Gamma_{t-1}]\nonumber 
    &=\mathbb{I}\{t\in\mathcal{T}_{\text{active}}\}\mathbb{E}[\mathcal{L}_{f_t,\bm{b}_t} (x_t,\bm{\lambda}_t)\mid\Gamma_{t-1}] + \nonumber\\
    &\qquad \qquad \qquad+\mathbb{E}[\mathbb{I}\{t\in\mathcal{T}_{\text{active}}\}\langle\bm{\lambda}_t,\bm{b}_t(x_t)-\bm{\rho}\rangle\mid\Gamma_{t-1}]\nonumber\\
    &=\mathbb{I}\{t\in\mathcal{T}_{\text{active}}\}\bar{\mathcal{L}}_{f,\bm{b}}(\bm{\lambda}_t\mid\mathcal{P})\nonumber + \\
    &\qquad \qquad \qquad+\mathbb{E}[\mathbb{I}\{t\in\mathcal{T}_{\text{active}}\}\langle\bm{\lambda}_t,\bm{b}_t(x_t)-\bm{\rho}\rangle\mid\Gamma_{t-1}].\label{ma}
    \end{align}   
    Thus, summing Equation~\eqref{ma} over $t$ and taking the expectation on both sides, we have: 
    \begin{align*}
        \sum_{t \in [T]} \mathbb{E}\left[ \mathbb{E} [f_t(x_t) \cdot \mathbb{I} \{t \in \mathcal{T}_{\text{active}}\} | \Gamma_{t-1}]\right] 
        &= \sum_{t \in [T]} \mathbb{E} [ \mathbb{I} \{t \in \mathcal{T}_{\text{active}}\} \bar{\mathcal{L}}_{f,\bm{b}}(\bm{\lambda}_t | \mathcal{P}) \nonumber \\ &\mkern130mu  + \mathbb{E}  [ \mathbb{I} \{t \in \mathcal{T}_{\text{active}}\} \langle \bm{\lambda}_t, \bm{b}_t(x_t) - \bm{\rho} \rangle | \Gamma_{t-1}] ]    \end{align*}
    By Tower Rule ($\mathbb{E}[\mathbb{E}[\cdot | \Gamma_{t-1}]] = \mathbb{E}[\cdot]$) and definition of $\mathbb{I} \{t \in \mathcal{T}_{\text{active}}\}$ we obtain: 
    \[\mathbb{E}\left[\sum_{t \in \mathcal{T}_{\text{active}}} f_t(x_t) \right] = \mathbb{E}\left[ \sum_{t \in \mathcal{T}_{\text{active}}}  \bar{\mathcal{L}}_{f,\bm{b}}(\bm{\lambda}_t | \mathcal{P})\right] + \mathbb{E}\left[ \sum_{t \in \mathcal{T}_{\text{active}}} \langle \bm{\lambda}_t, \bm{b}_t(x_t) - \bm{\rho} \rangle \right] .
    \]
    By convexity of the expected Lagrangian, with both sides of Jensen's inequality equal to zero when there are no active rounds, we can write:
    \begin{align}
    \mathbb{E}\left[ \sum_{t \in \mathcal{T}_{\text{active}}}f_t(x_t) \right] &= \mathbb{E}\left[\sum_{t \in \mathcal{T}_{\text{active}}} \bar{\mathcal{L}}_{f,\bm{b}}(\bm{\lambda}_t | \mathcal{P})\right] + \mathbb{E}\left[ \sum_{t \in \mathcal{T}_{\text{active}}} \langle \bm{\lambda}_t, \bm{b}_t(x_t) - \bm{\rho} \rangle \right] \nonumber \\
    &\geq \mathbb{E}\left[ |\mathcal{T}_{\text{active}}| \bar{\mathcal{L}}_{f,\bm{b}}(\bm{\bar{\lambda}}_{\mathcal{T}_{\text{active}}} | \mathcal{P})\right] + \mathbb{E}\left[ \sum_{t \in \mathcal{T}_{\text{active}}} \langle \bm{\lambda}_t, \bm{b}_t(x_t) - \bm{\rho} \rangle \right]. \nonumber
    \end{align}
    Combining this result with Inequality~\eqref{opt_bound}, we can bound the regret as:   
    \begin{align}
    R_T &= \text{OPT}_{stoc} - \mathbb{E}_{\bm{\gamma}}\left[\sum_{t \in \mathcal{T}_{\text{active}}} f_t(x_t) \right] \nonumber \\
    & \leq \mathbb{E}_{\bm{\gamma}} \left[ |\mathcal{T}_{\text{active}}| \cdot \bar{\mathcal{L}}_{f,\bm{b}}(\bm{\bar{\lambda}}_{\mathcal{T}_{\text{active}}} | \mathcal{P}) \right] + \mathbb{E}_{\bm{\gamma}} [ |\mathcal{T}_{\text{forced}}| ] \nonumber\\
    &\quad- \mathbb{E}\left[ |\mathcal{T}_{\text{active}}| \bar{\mathcal{L}}_{f,\bm{b}}(\bm{\bar{\lambda}}_{\mathcal{T}_{\text{active}}} | \mathcal{P})\right] \nonumber\\
    &\quad- \mathbb{E}\left[ \sum_{t \in \mathcal{T}_{\text{active}}} \langle \bm{\lambda}_t, \bm{b}_t(x_t) - \bm{\rho} \rangle \right] \nonumber\\
    & \leq \mathbb{E}_{\bm{\gamma}} [ |\mathcal{T}_{\text{forced}}| ] - \mathbb{E}\left[ \sum_{t \in \mathcal{T}_{\text{active}}} \langle \bm{\lambda}_t, \bm{b}_t(x_t) - \bm{\rho} \rangle \right] \nonumber \\
    &\leq \mathbb{E}_{\bm{\gamma}} [ |\mathcal{T}_{\text{forced}}| ] + \mathbb{E}\left[  - |\mathcal{T}_{\text{forced}}| + \frac{2}{\rho + \beta} + \frac{9 m}{\eta (\rho + \beta)^2} + 4 \eta T m \right] \label{rb} \\ 
    &\leq \frac{2}{\rho + \beta} + \frac{9 m}{\eta (\rho + \beta)^2} + 4 \eta T m. \nonumber
    \end{align} 
    where Inequality~\eqref{rb} holds by Lemma~\ref{aux2}. 
    To obtain the stated bounds, we tune the learning rate $\eta$ depending on whether the replenishment parameter $\beta$ is known to the algorithm: 
    \begin{itemize}
        \item Setting 1 ($\beta$ is known): $\eta \coloneqq \frac{3}{2 (\rho + \beta)} \sqrt{\frac{1}{T}} $. The regret is bounded by: 
            \begin{align*}
                R_T &\leq \frac{2}{\rho + \beta} + \frac{9 m}{\eta (\rho + \beta)^2} + 4 \eta T m \\
                &\leq \frac{2 + 12 m \sqrt{T}}{\rho + \beta}.
            \end{align*} 
        \item Setting 2 ($\beta$ is unknown): If the algorithm does not know the value of $\beta$, the Slater's margin can be lower bounded by $\rho$, since $\beta \geq 0$. Consequently, we can set $\eta \coloneqq \frac{3}{2 (\rho)} \frac{1}{\sqrt{T}}$. The regret is bounded by: 
            \begin{align*}
                R_T &\leq \frac{2}{\rho + \beta} + \frac{9 m}{\eta (\rho + \beta)^2} + 4 \eta T m \\ 
                &\leq \frac{2}{\rho + \beta} + \frac{6 m \rho}{(\rho + \beta)^2} \sqrt{T} + \frac{6 m}{\rho} \sqrt{T} \\
                &= \frac{2}{\rho + \beta} + 6m \left( \frac{\rho}{(\rho + \beta)^2} + \frac{1}{\rho} \right) \sqrt{T}.
            \end{align*}         
    \end{itemize}
    This concludes the proof.
\end{proof}

\advr*
\begin{proof}
    If $\eta>\nicefrac{1}{2m}$, both claims follow from $\alpha\text{-}R_T\leq T$ by the same argument as in the proof of Theorem~\ref{stre}. Thus, assume $0<\eta\leq\nicefrac{1}{2m}$.
    Let $(x_t^*)_{t=1}^T$ be the optimal sequence of actions and let 
    $\alpha := \frac{\rho + \beta}{1 + \beta}$.
    For each round $t \in [T]$, we define $\xi_t^\diamond$ as the strategy that plays $x_t^*$ with probability $\alpha$ and $\varnothing$ with probability $1 - \alpha$.
    By construction, for each $i \in [m]$, it holds: 
    \begin{align}
        \mathbb{E}_{x\sim \xi_t^\diamond}[b_{t,i}(x)] = \alpha \cdot b_{t,i}(x_t^*) + (1 - \alpha) b_{t,i}(\varnothing) \leq \alpha \cdot 1 + (1 - \alpha)(-\beta) = \rho.\label{alpha_b}
    \end{align}
    Moreover, since rewards are non-negative, it holds: 
    \begin{equation} \mathbb{E}_{x\sim \xi_t^\diamond}[f_t(x)] = \alpha \cdot f_t(x_t^*) + (1 - \alpha) f_t(\varnothing) \geq \alpha \cdot f_t(x_t^*). \label{alpha_f} \end{equation} 
Exploiting these properties, we can bound the $\alpha$-regret as: 
\begin{align}
    \alpha\text{-}R_T &=\sum_{t \in [T]} ( \alpha  f_t(x_t^*) - f_t(x_t)) \nonumber \\ 
    &\leq \sum_{t \in [T]} (\mathbb{E}_{x\sim\xi_t^\diamond}[f_t(x)] -f_t(x_t) ) \label{exp} \\ 
    &= \sum_{t \in \mathcal{T}_{\text{active}}} (\mathbb{E}_{x\sim\xi_t^\diamond}[f_t(x)] -f_t(x_t) ) + \sum_{t \in \mathcal{T}_{\text{forced}}} (\mathbb{E}_{x\sim\xi_t^\diamond}[f_t(x)] -f_t(\varnothing) ) \nonumber \\ 
    &\leq \sum_{t \in \mathcal{T}_{\text{active}}} (\mathbb{E}_{x\sim\xi_t^\diamond}[f_t(x)] -f_t(x_t) ) + |\mathcal{T}_{\text{forced}}| \nonumber \\
    &\leq \sum_{t \in  \mathcal{T}_{\text{active}}} \left(\langle \bm{\lambda}_t, \mathbb{E}_{x\sim\xi_{t}^\diamond}[\bm{b}_t(x) -\bm{\rho}] - (\bm{b}_t(x_t) - \bm{\rho}) \rangle\right) + |\mathcal{T}_{\text{forced}}| \label{reg_dec} \\
    &\leq - \sum_{t \in  \mathcal{T}_{\text{active}}} \langle \bm{\lambda}_t, \bm{b}_t(x_t) -\bm{\rho} \rangle + |\mathcal{T}_{\text{forced}}|  \label{optA} \\ 
    &\leq - |\mathcal{T}_{\text{forced}}| + \frac{2}{\rho + \beta} + \frac{9 m}{\eta (\rho + \beta)^2} + 4 \eta T m + |\mathcal{T}_{\text{forced}}| \label{a2} \\
    &= \frac{2}{\rho + \beta} + \frac{9 m}{\eta (\rho + \beta)^2} + 4 \eta T m, \nonumber
    \end{align}
    where Inequality~\eqref{exp} follows from Inequality~\eqref{alpha_f}, Inequality~\eqref{reg_dec} holds since $x_t$ deterministically maximizes the instantaneous Lagrangian over $\mathcal{X}$, Inequality~\eqref{optA} follows from Inequality~\eqref{alpha_b} and Inequality~\eqref{a2} holds by Lemma~\ref{aux2}. 
    The proof is concluded by tuning $\eta$ according to the reasoning in Theorem~\ref{stre}.
\end{proof}

\end{document}